\documentclass[acmsmall,screen,nonacm]{acmart}
\usepackage{booktabs}
\usepackage{tabularx}
\usepackage{amsmath}
\usepackage{algorithm}
\usepackage{algpseudocode}
\newtheorem{proposition}{Proposition}

\newcommand{\AP}{\mathrm{AP}}
\newcommand{\MAP}{\mathrm{MAP}}
\acmDOI{}
\acmISBN{}
\setcopyright{none}
\hypersetup{keeppdfinfo=true}
\date{}
\makeatletter
\def\@ACM@checkaffil{%
  \if@ACM@instpresent\else
    \ClassWarningNoLine{acmart}{No institution present for an affiliation}%
  \fi}
\makeatother
\begin{document}
\title[World Requirement Model]{World Requirement Model: Learning Requirement-Change Consequences from Typed Artifact Graphs}
\author{Yuanpeng He}
\affiliation{\institution{Peking University}}
\email{heyuanpeng@stu.pku.edu.cn}

\author{Lijian Li}
\affiliation{\institution{Macau University}}
\email{yc57484@um.edu.mo}

\author{Dongming Jin}
\affiliation{\institution{Peking University}}
\email{dmjin@stu.pku.edu.cn}

\author{Huanyao Zhang}
\affiliation{\institution{Peking University}}
\email{2301111978@stu.pku.edu.cn}

\author{Fangjing Li}
\affiliation{\institution{Beijing Jiaotong University}}
\email{lifangjing@bjtu.edu.cn}

\author{Linyu Li}
\affiliation{\institution{Peking University}}
\email{linyuli@stu.pku.edu.cn}

\author{Chung-ju Huang}
\affiliation{\institution{Peking University}}
\email{chongruhuang.pku@gmail.com}

\author{Tianxiang Zhan}
\affiliation{\institution{Peking University}}
\email{zhantianxianguestc@hotmail.com}

\author{Qingsong Wen}
\affiliation{\institution{Squirrel Ai Learning}}
\email{qingsongedu@gmail.com}

\author{Wenpin Jiao}
\affiliation{\institution{Peking University}}
\email{jwp@pku.edu.cn}

\renewcommand{\shortauthors}{Yuanpeng He et al.}
\begin{abstract}
Requirement changes can affect connected stakeholders, constraints, components, and tests. We present World Requirement Model (WRM), which encodes this engineering context as a typed artifact graph and predicts consequences at shared artifact identifiers. Relation-aware attention and typed propagation contextualize nodes; world and decision representations support learned dynamics. Shared readouts score impact, conflict, violation, and defect risk; auxiliary objectives supervise successor adjacency and latent prediction. On 28 scored cases from 282 synthetic cases in six domains, WRM obtains impact mean average precision (MAP) of 0.724 versus 0.623 for a hashed-text multilayer perceptron (MLP), a 16.2\% relative gain and paired difference of 0.101 (conditional 95\% interval [0.044,0.159]). Lowest-quarter mean AP improves by 32.4\%, and equal-domain MAP by 13.3\%. Four 47-case comparisons on an expanded corpus show MAP gains of 15.6--29.1\% and higher means on all five reported metrics. The recorded advantage thus extends across score summaries and annotation/training settings. Checkpoints were selected on scored cases, and backbones are unmatched, so these results characterize selected systems. Our analysis establishes candidate-coverage bounds and shows that the current linear impact head cannot rerank a fixed world's artifacts across decisions. WRM contributes an artifact-addressed world-model formulation, comparative evidence for contextual consequence scoring, and explicit conditions for evaluating requirement-world prediction.
\end{abstract}
\begin{CCSXML}
	<ccs2012>
	<concept>
	<concept_id>10011007.10011074.10011075.10011076</concept_id>
	<concept_desc>Software and its engineering~Requirements analysis</concept_desc>
	<concept_significance>500</concept_significance>
	</concept>
	<concept>
	<concept_id>10010147.10010257.10010321.10010336</concept_id>
	<concept_desc>Computing methodologies~Feature selection</concept_desc>
	<concept_significance>500</concept_significance>
	</concept>
	</ccs2012>
\end{CCSXML}

\ccsdesc[500]{Software and its engineering~Requirements analysis}
\ccsdesc[500]{Computing methodologies~Feature selection}
\keywords{requirements engineering, change impact analysis, typed graphs, world models, decision-conditioned prediction, evaluation validity}
\maketitle
\hypersetup{
  pdfauthor={Yuanpeng He, Lijian Li, Dongming Jin, Huanyao Zhang, Fangjing Li, Linyu Li, Chung-ju Huang, Tianxiang Zhan, Qingsong Wen, Wenpin Jiao},
  pdfkeywords={requirements engineering, change impact analysis, typed graphs, world models, decision-conditioned prediction, evaluation validity}
}

\section{Introduction}\label{sec:intro}
A short requirement edit can create review obligations across an artifact network. Adding live-video monitoring to an examination system raises questions about affected students and instructors, tracking dependencies, and consent constraints. An engineer must connect the edit to these artifacts before deciding what to review. Requirements engineering (RE) supplies domain assumptions, stakeholder goals, and specifications~\cite{zave1997,lamsweerde2009}; traceability makes their connections available for change analysis~\cite{gotel1994traceability,cleland2012,goknil2016impact}. The prediction target extends beyond classifying the edited requirement to identifying relevant artifacts and distinguishing their potential consequences.

Natural language processing (NLP) supports classification, concept identification, and trace recovery~\cite{zhao2020nlp}; large language model (LLM) studies address elicitation and classification~\cite{ronanki2023,alhoshan2025}; rule-based analysis propagates changes through explicit relations~\cite{goknil2016impact}. WRM connects several consequence questions to a shared artifact graph. An artifact's role depends on what constrains, implements, or verifies it, as well as its text. Preserving those typed connections makes context available to prediction while keeping each score attached to an artifact an engineer can inspect.

World models motivate a state--action--successor formulation~\cite{ha2018world,hafner2019dream}. In RE, the state records engineering knowledge and the action specifies an artifact edit. The \emph{World Requirement Model} (WRM) instantiates this formulation through language encoding, relation-aware attention, and typed propagation. Shared readouts attach impact, conflict, violation, and defect-risk scores to artifact identifiers; auxiliary objectives predict successor adjacency and a successor representation. Separating contextual artifact features from global decision context is central to both the architecture and the interpretation of its predictions.

We study 282 synthetic cases in six domains and records associated with a 475-case corpus. On the paired 28-case cohort, WRM improves MAP by 16.2\% over the strongest control, a hashed-text MLP, and lowest-quarter mean AP by 32.4\%. A 13.3\% gain under equal domain weights and gains in all four expanded-corpus settings strengthen the descriptive evidence beyond one average. WRM still loses eight paired cases; checkpoints used the scored cases and backbones are unmatched. The study therefore tests contextual scoring and its robustness while auditing the conditions required for broader world prediction.

Four research questions structure the study: \textbf{RQ1}, how WRM compares with available controls for consequence scoring; \textbf{RQ2}, which candidate and transition claims the evaluation supports; \textbf{RQ3}, how comparisons vary with world structure and configuration; and \textbf{RQ4}, whether gains persist across the score distribution and within worlds.

Our contributions are (1) a requirement-world formulation that combines typed artifact representations with shared, artifact-addressed consequence readouts; (2) comparative evidence spanning impact ranking, constraint scoring, expanded-corpus settings, and distributional robustness; and (3) candidate-coverage and lower-tail bounds, together with a scoring-function analysis that makes decision dependence explicit. These contributions connect a concrete prediction architecture, measured scoring advantages, and testable requirements for valid world prediction.

\section{Problem and Motivation}\label{sec:problem}
\subsection{Requirement context and candidate decisions}
A requirement world is a partial record of the engineering context supplied to the model. We represent it as a typed directed multigraph
\begin{equation}
 W_t=(V_t,E_t,\tau_V,\tau_E,X_t),
\end{equation}
where $V_t$ contains artifact nodes, $E_t$ contains relations (possibly several types for one ordered node pair), $\tau_V$ and $\tau_E$ assign types, and $X_t$ associates textual descriptions with nodes. The implementation has eight node types: requirement, stakeholder, constraint, assumption, entity, component, test, and glossary. Its ten relation types are dependency, conflict, refinement, composition, trace, effect, authorship, constraint, satisfaction, and verification. This engineering vocabulary is motivated by the distinction between domain knowledge and specifications~\cite{zave1997}. Types structure the supplied context \mbox{without guaranteeing completeness or correctness}.

A candidate decision $d_t=(o_t,\iota_t,p_t)$ contains an operation, a target identifier, and a payload. The schema supports addition, modification, deletion, linking, unlinking, splitting, merging, prioritization, and elicitation. Schema coverage does not imply empirical coverage: almost all decisions in the examined corpus are additions (Section~\ref{sec:data}). A payload can introduce a node that does not yet occur in $V_t$. This detail is consequential because a newly introduced requirement may itself participate in the conflicts to be predicted.

The target consequence object is
\begin{equation}
 C_{t+1}=(\mathcal C,\mathcal I,\mathcal V,\mathcal G,\mathcal D),
\end{equation}
with conflict pairs $\mathcal C$, ranked impacted artifacts $\mathcal I$, violated constraints $\mathcal V$, requirement gaps $\mathcal G$, and downstream defect-risk annotations $\mathcal D$. In the synthetic data, these objects express generated analyst judgments. They do not document failures observed after software deployment. We use ``defect risk'' to retain that distinction.

The broader modeling objective is to approximate
\begin{equation}
 \widehat T_\theta(W_t,d_t)=(\widehat W_{t+1},\widehat C_{t+1}).\label{eq:objective}
\end{equation}
Equation~\ref{eq:objective} states the world-model objective. The evaluated decoder scores consequences and adds relations over supplied nodes; Section~\ref{sec:method} defines its output limits and decision dependence.

\subsection{A motivating case from the corpus}
Figure~\ref{fig:worldedit} follows examination-system case \texttt{dd91978a6d-0}. Adding live-video requirement \texttt{r4} links affected students and instructors to gaze-tracking and consent obligations. The generated consequence labels flag the consent constraint as a potential violation and the tracking component as a defect-risk target. These are synthetic judgments, not observed failures or a verified privacy assessment. The example shows why a common world representation should identify relevant artifacts while distinguishing consequence types; it does not demonstrate model accuracy.

\begin{figure}[t]
\centering\includegraphics[width=\linewidth]{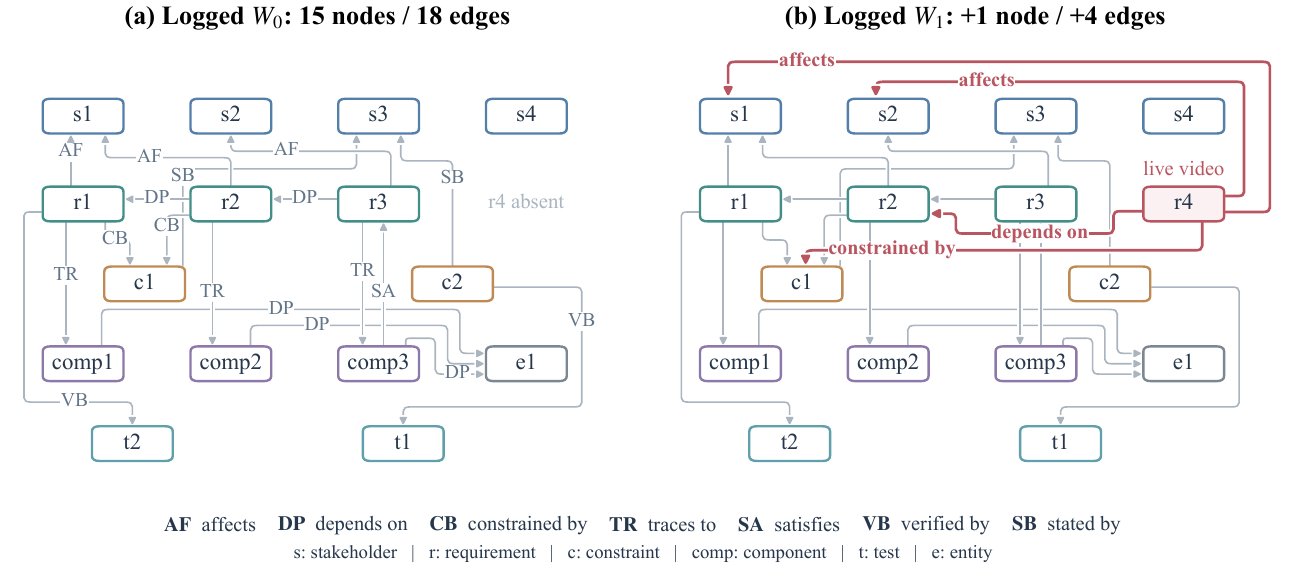}
\caption{Typed relations in a recorded synthetic world edit. For example, $u\xrightarrow{\mathrm{DP}}v$ means ``$u$ depends on $v$''.  Labels in (a) identify the 18 gray relations retained in both panels. In (b), new live-video requirement r4 affects student s1 and instructor s2, depends on gaze-tracking requirement r2, and is constrained by consent c1. The reverse links r3$\to$comp3 (TR) and comp3$\to$r3 (SA) express traceability and satisfaction, respectively. Crossings are not junctions; bends and lengths serve only layout. $W_1$ records a synthetic supervision target rather than an output predicted by WRM.}
\label{fig:worldedit}
\Description{The left graph labels all eighteen retained directed relations with seven decoded type abbreviations. The right graph retains them at identical positions and names four red additions: r4 affects s1 and s2, depends on r2, and is constrained by c1. Reverse r3/component links denote tracing and satisfaction, not feedback flow.}
\end{figure}

Existing-artifact ranking uses current nodes $V_t$, whereas addition conflicts need the new node as a candidate. A decoder can prioritize existing artifacts while failing to express new conflicts. Evaluation must distinguish these output universes.

\begin{table}[t]
\caption{Interfaces of evaluated implementations. ``Shift'' refers to the node-independent logit offset induced by a decision (Section~\ref{sec:advantage}).}
\label{tab:interfaces}
\centering\small
\begin{tabularx}{\linewidth}{l>{\raggedright\arraybackslash}X>{\raggedright\arraybackslash}X>{\raggedright\arraybackslash}X}
\toprule
System & Artifact representation & Decision access & Consequence interface\\
\midrule
Hashed-text MLP & Independent node text & None & Node scores; Conflict endpoint rule\\
Type/degree LR & Node type and degree & None & Separate node classifiers\\
World-text 1-NN & Pooled world text & None & Transfer source annotations\\
TF-IDF & Lexical node features & Text query & Similarity ranks; pair rule\\
API prompt & Serialized nodes/edges & Supplied payload & Generated consequence JSON\\
WRM & Learned text and typed graph & Shift for node scores & Shared graph; learned pair/node heads\\
\bottomrule
\end{tabularx}
\end{table}

Table~\ref{tab:interfaces} contrasts the interfaces. WRM and prompted methods both receive relational context; WRM represents it through shared node features and native scoring heads. The comparison concerns how each instantiated system uses that context.

\subsection{An evaluation contract}
To interpret a reported score, we check the state--decision pair, the head's output universe, and label provenance through three conditions.

\textbf{Candidate closure.} Every positive must be representable or counted as outside the decoder's scope. Additions may require materializing candidate nodes before conflict scoring. Removing out-of-universe positives without reporting coverage changes the evaluated task.

\textbf{Action-label alignment.} Annotations must concern the recorded decision $d_t$. If annotators propose different edits, their consequences are different conditional targets. Aggregating them does not establish agreement about one state--decision pair.

\textbf{Selection separation.} Fitting, checkpoint selection, threshold tuning, and final assessment must use separate information. Disjoint trajectory identifiers do not make a selection partition an untouched test set~\cite{cawley2010,kapoor2023}. This distinction bounds the interpretation of reported intervals.

\begin{proposition}[Candidate-set ceiling]\label{prop:ceiling}
Let $U$ be the artifact identifiers a decoder can emit, $G_I$ a nonempty set of relevant impact identifiers, and $G_C$ a nonempty set of relevant unordered conflict pairs. For a duplicate-free list of at most $k$ predictions, average precision normalized by $|G_I|$ is at most $\min(k,|G_I\cap U|)/|G_I|$. Conflict recall is at most $|G_C\cap\binom{U}{2}|/|G_C|$.
\end{proposition}
\begin{proof}
Each retrieved relevant artifact contributes at most one to the AP numerator, and at most $\min(k,|G_I\cap U|)$ such artifacts can be emitted. A conflict decoder can recover only pairs with both endpoints in $U$. Dividing these maximal hit counts by the respective positive-set sizes gives the bounds.
\end{proof}

\section{The World Requirement Model}\label{sec:method}
\begin{figure}[t]
\centering
\includegraphics[width=\linewidth]{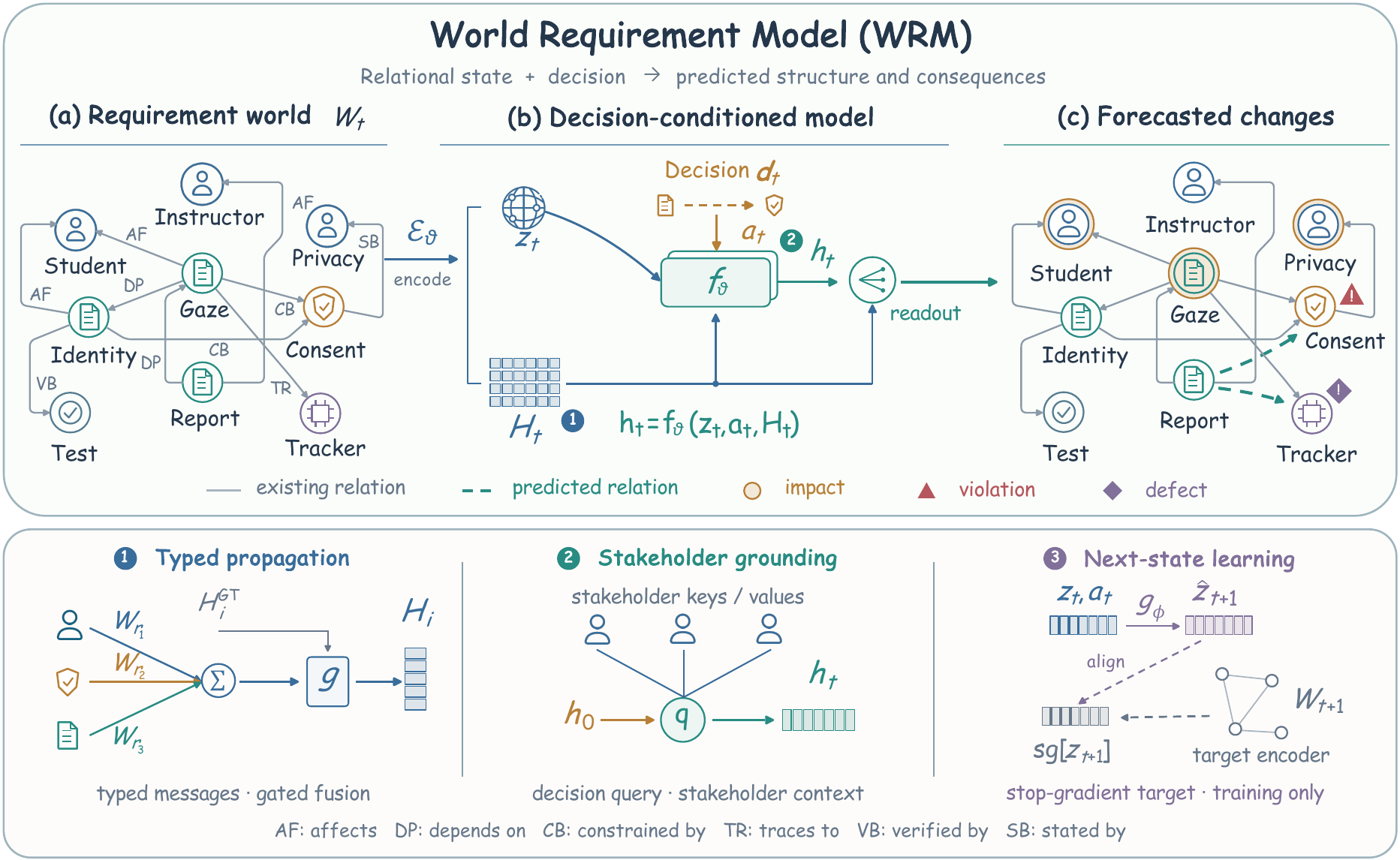}
\caption{World Requirement Model. Panels (a,c) retain nine artifacts and ten typed relations from Figure~\ref{fig:worldedit}; Privacy denotes Privacy Officer. Gray arrows use the bottom key; crossings are not junctions. Insets 1--3 detail propagation, stakeholder grounding, and training-only latent alignment. Colored additions in (c) are schematic: ``Forecasted changes'' and ``predicted relation'' name output types, not measured results.}
\label{fig:overview}
\Description{Nine-node requirement graphs flank world encoding, decision conditioning and shared readout. Gray arrows show ten recorded typed relations; green arrows and consequence marks illustrate output vocabulary. Three lower insets depict typed affine propagation and gated fusion, stakeholder attention, and latent prediction aligned to a stop-gradient successor encoding.}
\end{figure}
Figure~\ref{fig:overview} uses $\mathcal E_\theta(W_t)=(z_t,H_t)$ for world encoding and $h_t=f_\theta(z_t,a_t,H_t)$ for dynamics followed by stakeholder grounding. Node identities attach readouts to supplied artifacts; decoding limits the targets and edits that can be emitted. Solid model arrows denote forward computation; dashed alignment and target-encoding paths belong only to training. The readout symbol groups heads with different inputs. In particular, the conflict head omits $h_t$, so its pairwise scores require separate interpretation from decision-conditioned node scores. Inset 3 uses a supplied successor $W_{t+1}$ as supervision, not another forecast generated by $f_\theta$.

\subsection{Text and graph representation}
Qwen2.5~\cite{qwen2024} with low-rank adaptation (LoRA)~\cite{hu2021lora} encodes node text. Masked mean pooling, projection, and a type embedding yield
\begin{equation}
 h_i^{\mathrm{in}}=P\,e_\mathrm{text}(x_i)+e_\mathrm{type}(\tau_V(i)).
\end{equation}
Following graph attention and structural-bias modeling~\cite{velickovic2018gat,ying2021graphormer}, we prepend a learned world token and bias attention by relation type. With query/key projections $q_i,k_j$ and head width $d_h$:
\begin{equation}
 \alpha_{ij}\propto\exp\left(\frac{q_i^\top k_j}{\sqrt{d_h}}+\beta_{ij}\right).
\end{equation}
For each head, $\beta_{ij}$ sums learned biases for all relation types from $i$ to $j$, with a no-edge bias otherwise; the head index is suppressed. The world token yields $z_t$; node features remain addressable. Attributes enter only through text or types. Losses mask padding, but attention leaves padded logits unmasked, potentially affecting graph-size comparisons.

\subsection{Typed propagation and decision conditioning}
The consequence-propagation layer (CPL) aggregates graph-attention features along typed relations. With row-normalized adjacency $A^{(r)}$ for relation $r$ and node-feature rows in $H^{(0)}=H^{\mathrm{GT}}$, two rounds reuse row-wise affine maps $\Phi_r$ and $\Phi_{\mathrm{self}}$:
\begin{align}
 M^{(\ell)}&=\Phi_{\mathrm{self}}(H^{(\ell)})+\sum_r\Phi_r(A^{(r)}H^{(\ell)}),\\
 H^{(\ell+1)}&=\operatorname{GELU}(M^{(\ell)}).
\end{align}
Each $\Phi_r$ has weights and bias; inset 1 labels its weight matrix $W_r$ and the fusion gate $g$, distinct from latent predictor $g_\phi$. The illustrated $r_1,r_2,r_3$ index relation channels; self-channel, bias, GELU, and repeated-round details are omitted. Gated fusion combines propagated and graph-attention features before dropout. Decision embedding $a_t$ combines the operation with existing target text, otherwise payload text or the operation name; modifications may therefore omit the new payload. An MLP fuses $z_t,a_t$; residual dynamics yield $h_t^0$. Inset 2 abbreviates this as $h_0$ and labels its attention query $q$. Stakeholder-grounded consequence heads (SGCH) attend to stakeholder features $H_S$:
\begin{equation}
 u_t=\operatorname{MHA}(h_t^0,H_S,H_S),\qquad
 h_t=h_t^0+\tanh F([h_t^0,u_t]).
\end{equation}
Stakeholder attention uses four heads, falling back to all valid nodes when no stakeholders exist. Attention supplies context, not validated responsibility attributions. The world vector $z_t$ precedes CPL, while $H_t$ comprises post-CPL node rows $H_i$: the former supplies dynamics, and the latter preserves artifact identities for readouts.

\subsection{Consequence heads and decoding}
Impact, violation, and defect heads score $[H_i,h_t]$ for each node. The conflict head scores pairs using $[H_i,H_j,H_i-H_j,H_i\odot H_j]$, without direct decision input. These distinct inputs matter for the action-dependence analysis in Section~\ref{sec:advantage}. The decoder emits violation scores only for constraints and defect scores only for requirements, components, and tests (Table~\ref{tab:configuration}). Omitted defect scores become zero, so the reported AUC evaluates emitted predictions, including thresholding and type filters. The gap head predicts a text embedding; inference exposes only a generic signal when its norm exceeds 1.0, rather than a semantically identified missing requirement.

\subsection{Training and auxiliary transition objectives}
Requirement-graph contrastive pretraining pairs a text-masked view with an edge-dropped view of each world, using other worlds in the batch as negatives. It encourages agreement between these views of the same context before consequence supervision. Fine-tuning shares the encoder across pair/node classification, successor-adjacency reconstruction, latent prediction, listwise impact supervision, and gap-embedding alignment. Table~\ref{tab:configuration} gives the recorded settings. The saved backbone tag names Qwen2.5-7B, while local loading resolves an instruction-tuned snapshot; exact snapshot identity remains necessary for complete neural reproduction.

\begin{table}[t]
\caption{Recorded configuration and decoding contract. These are implementation settings, not hyperparameter sweeps or measured runtime advantages.}
\label{tab:configuration}
\centering\small
\begin{tabularx}{\linewidth}{lXlX}
\toprule
Component & Setting & Component & Setting\\
\midrule
Node text & 96 tokens; mean pooling & Graph width & 256 dimensions\\
Graph attention & 3 layers; 8 heads & Typed propagation & 2 rounds; 10 relations\\
Graph capacity & 48 padded nodes & LoRA & Rank 16; scale 32\\
Pretraining & 2 contrastive epochs & Fine-tuning & 10 epochs; batch 4\\
Fine-tune opt. & AdamW; LR $2\times10^{-4}$ & Gradient clip & 1.0\\
Impact output & Top 10 current nodes & Pair/violation & Threshold 0.5\\
Refinement & 3 supervised training steps & Headline path & Base forward inference\\
\bottomrule
\end{tabularx}
\end{table}

Training includes three supervised refinement steps. The headline drivers use the base forward wrapper, while a separate driver restores refinement.

A nonlinear edge decoder combines node-pair features with the decision context, while a separate latent predictor uses the world and decision embeddings:
\begin{equation}
 \ell_{ij}^{\mathrm{edge}}=D_\omega([H_i,H_j,h_t]),\qquad
 \widehat z_{t+1}=g_\phi([z_t,a_t]).
\end{equation}
Edge reconstruction targets successor adjacency; latent alignment instead predicts a successor representation, following the image-based joint-embedding predictive architecture (I-JEPA)~\cite{assran2023ijepa}. WRM uses successor graphs without image masking or a momentum teacher. With $\bar z=z/\|z\|_2$ and a stop-gradient successor target, \mbox{the implemented alignment objective is}
\begin{equation}
 L_z=2-2\langle\overline{\widehat z}_{t+1},\operatorname{sg}[\bar z_{t+1}]\rangle.\label{eq:latent_alignment}
\end{equation} The same world encoder supplies the stop-gradient target. Inset 3 shows latent vectors before the normalization in Equation~\ref{eq:latent_alignment}. Graph decoding selects the highest-scoring relation type per current-node pair and adds it at score 0.5 or above. Input nodes and edges remain intact, permitting relation additions but not node/text edits or edge deletion. Supervision additionally requires usable successor graphs aligned with the recorded decision.

\textbf{Within-world refinement.} Iterative consequence reasoning (ICR) encodes $W_t$ once and keeps $H_t,z_t,a_t$ fixed. A five-dimensional summary contains the mean and variance of impact logits and the means of conflict, violation, and defect logits. After the initial prediction, the preceding summary $s^{(k-1)}$ enters a gated update $h_t^{(k)}=R(h_t,s^{(k-1)})$ together with the base context. Heads are supervised at every step. The three steps refine one decision. Rollout instead requires decoding and re-encoding successors under successive decisions; P evaluates only the former.

\subsection{Representational advantages and scoring semantics}\label{sec:advantage}
\textbf{Relational context.} WRM distinguishes artifacts through both content and engineering relationships. Identical descriptions can receive different representations when verification links, constraints, or neighboring artifacts differ. Relation-specific aggregation preserves those distinctions~\cite{battaglia2018relational,schlichtkrull2017}, whereas independent text scores and degree counts omit them. This is a representational advantage: expressive capacity depends on the aggregation scheme~\cite{xu2019gnn}, and matched ablations must establish its contribution to accuracy.

\textbf{Shared supervision and artifact addressing.} Related objectives update a common encoder~\cite{caruana1997}, so one artifact can inform several consequence questions without losing its identity. A constraint can receive both an impact priority and a violation score, linking what to review with the obligation at risk. Direct scoring of current artifacts also avoids the cross-world identifier matching required by label transfer, which the implemented retrieval control does not verify. These are architectural properties; runtime and analyst-effort benefits remain unmeasured. The heads preserve distinct semantics: high impact need not imply violation or defect risk. Cross-head checks must therefore respect output universes and budgets, since a flagged artifact need not occur in every truncated list.

\textbf{Ranking-aligned supervision.} The impact listwise objective treats a world as a candidate list~\cite{cao2007}. For a nonempty set $G^+$ of positive current nodes and logits $s_i$, the loss is
\begin{equation}
 L_{\mathrm{rank}}=-\frac{1}{|G^+|}\sum_{i\in G^+}
 \log\frac{\exp(s_i)}{\sum_{j\in V}\exp(s_j)}.\label{eq:listwise}
\end{equation}
Worlds without positive current targets contribute zero to this loss. Equation~\ref{eq:listwise} supplies list-level supervision, not a differentiable MAP objective; a matched ablation must establish its contribution.

\textbf{Scope of action conditioning.} The current linear impact head factorizes as
\begin{equation}
 s_i(W,d)=\langle u,H_i(W)\rangle+\langle v,h(W,d)\rangle+b.\label{eq:factorization}
\end{equation}
At inference with fixed parameters and $W$, $s_i-s_j=\langle u,H_i-H_j\rangle$: the decision term cancels from ordering and the listwise softmax. Decisions can shift sigmoid scores and thresholded outputs, but cannot rerank these features. The impact comparison therefore tests world-context scoring. Action-specific reranking needs node--decision interactions and evaluation on fixed worlds with different edits; the nonlinear edge decoder requires its own transition assessment.

\section{Data and Study Design}\label{sec:data}
\subsection{Corpus construction and provenance}
WRM-Bench uses an LLM-based protocol to generate typed worlds, edits, and consequences from domain scenarios. The recorded endpoint label is \texttt{glm-5.2}. Generation emphasizes high-impact decisions and stakeholder tensions, providing inspectable graph inputs rather than an estimate of industrial consequence prevalence.

Relabeling uses majority votes for conflicts, violations, and defects, retains ten frequent impact targets, and deduplicates gaps. Its prompt omits the stored decision and requests a new high-impact edit, then attaches the generated consequences to the original decision. Raw repeated decisions are unavailable, preventing estimation of the mismatch rate. The labels therefore have uncertain action alignment: interpretable self-consistency requires fixed state--decision inputs~\cite{wang2022consistency}.

\begin{table}[t]
\caption{Primary corpus statistics from the frozen \texttt{v7\_rich} snapshot. Partition names follow the data file; the training driver combines Train and Validation for fitting and uses Nominal test for checkpoint selection.}
\label{tab:data}
\centering
\begin{tabular}{lrrrr}
\toprule
Characteristic & Train & Validation & Nominal test & Total\\
\midrule
Decision cases & 227 & 27 & 28 & 282\\
Mean input nodes & 15.37 & 15.44 & 15.04 & 15.34\\
Mean input edges & 17.33 & 16.89 & 17.29 & 17.28\\
Empty successor graphs & 222 & 27 & 27 & 276\\
Conflict annotations & 432 & 62 & 56 & 550\\
Conflicts within input nodes & 5 & 2 & 0 & 7\\
Impact annotations outside input & 197 & 22 & 27 & 246\\
\bottomrule
\end{tabular}
\end{table}

The snapshot contains 282 cases from 277 trajectories: 272 have one step and five have two. Of the decisions, 279 are additions, with one each of elicitation, linking, and modification. Domain counts are automotive 82, e-commerce 51, fintech 49, health 48, aerospace 42, and education 10. Worlds contain 15.34 nodes and 17.28 edges on average (Table~\ref{tab:data}). These observed trajectories define the available temporal evidence, irrespective of the generator's requested maximum length.

\begin{figure}[t]
\centering\includegraphics[width=\linewidth]{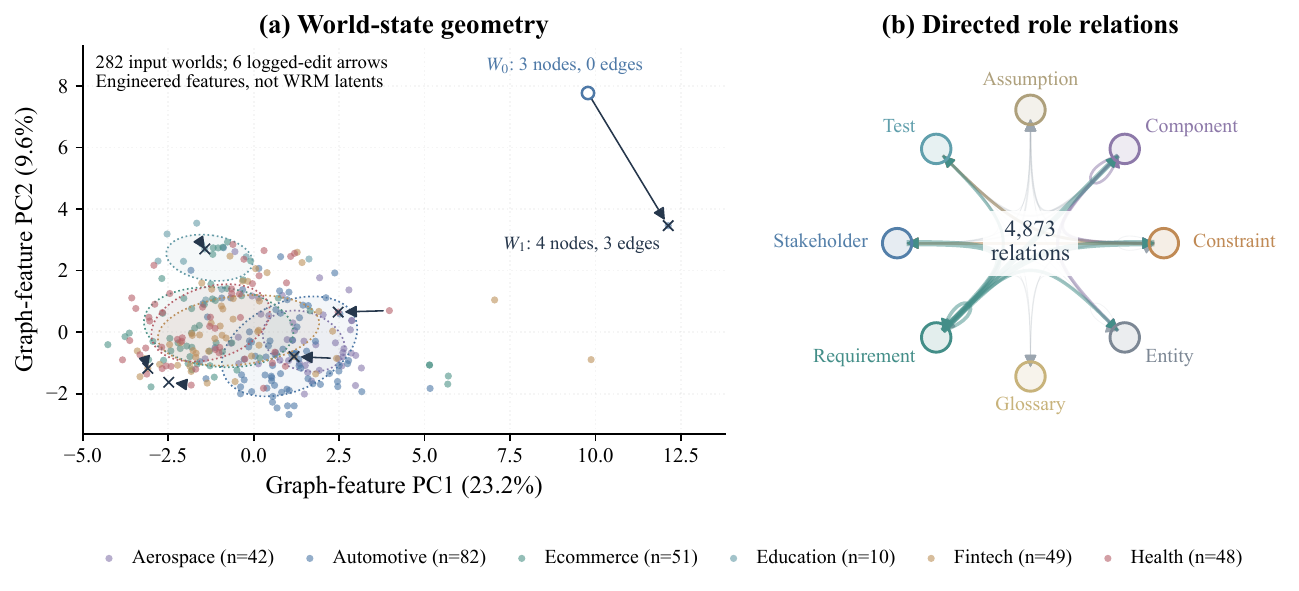}
\caption{Requirement-world structure. (a) Graph-feature PCA of 282 input worlds with six logged-successor arrows; ellipses describe covariance, not confidence. Coordinates are engineered descriptors, not WRM latents. (b) Directed chords aggregate 4,873 relations among eight roles; widths use square-root scaling, with the 12 largest role pairs emphasized.}
\label{fig:corpus}
\Description{A graph-feature state-space projection displays domains and logged transitions; a directed chord network shows aggregate relations among artifact roles.}
\end{figure}

Figure~\ref{fig:corpus} characterizes the supplied worlds. PCA uses 21 descriptors: node/edge-type proportions, log counts, and density. It is fitted on input worlds, then projects the six nonempty successors without refitting. The distant automotive case \texttt{44b3018bf7} starts with three stakeholders and no edges, then adds a requirement and three relations; those unusual descriptors explain its position. The chord diagram summarizes connectivity across artifact roles. These views describe structural diversity and recorded edits, not learned latent dynamics or prediction error.

\subsection{Evaluation records and selection protocol}
We distinguish \textbf{C}, the ten-system comprehensive comparison; \textbf{F}, a seven-system API-focused comparison; and \textbf{E}, aligned case-level scores for paired analysis. C/F share 28 scored cases, but identically named methods can differ across records, so scores are never pooled. The v8 driver fits on 254 non-test cases, including 27 validation-labelled cases, and selects its checkpoint on the 28 test-labelled cases. The comparisons reuse selection cases.

\textbf{Expanded-corpus settings.} Four annotation/training settings each score 47 of 475 cases; 428 enter fitting and 47 select checkpoints. \emph{Original labels} and \emph{Revised labels} use CPL and stakeholder grounding without graph contrastive pretraining. \emph{Retrieval training} uses retrieved case context and learned task weights; \emph{Edge modulation} trains with decision-dependent relation biases. These names describe training conditions: the common evaluation loader restores only the base scorer, omitting separately saved auxiliary modules. They therefore do not isolate component effects. Aggregate summaries lack per-case prediction identifiers and remain separate from C/F/E.

\subsection{Comparators and comparison scope}
The C/F records contain 11 distinct control configurations: TF-IDF, BM25, SBERT, world-text 1-NN, hashed-text MLP, type/degree logistic regression (LR), GLM-5.2 prompting, GLM-5.3 zero-shot and CoT-labelled prompting, ICL-labelled prompting, and a keyword heuristic. Table~\ref{tab:allcomparators} includes every recorded configuration or identifies its identical repeat. Closed-model names are the endpoint labels used in the artifacts; the comparison concerns those instantiated pipelines.

Implementation identity is essential to interpreting this comparison. The PointwiseLoRA-labelled control is a hashed-text encoder with an MLP, without Qwen or LoRA. Graph-struct uses node types and in/out degrees with LR. World-text 1-NN and SBERT transfer retrieved-world labels without a general artifact-correspondence procedure; SBERT uses all-MiniLM-L6-v2~\cite{reimers2019}. GLM-5.2 receives a short example and a reasoning instruction~\cite{wei2022cot}. The GLM-5.3 zero-shot driver omits the decision payload and truncates nodes/edges; its CoT-labelled driver additionally omits node text, edges, and fitted demonstrations. ICL examples are unfitted, and the heuristic has an execution-error path. These failures remain visible in Table~\ref{tab:allcomparators}, but do not establish inferior method families. Learned text/structure controls use scikit-learn~\cite{pedregosa2011}. Decoding also differs: WRM emits up to ten current nodes above $10^{-3}$, while learned text/structure controls threshold at 0.5.

\begin{table}[t]
\caption{Direct consequence comparisons with complete C/F coverage. C: comprehensive run; F: API run. Higher is better. Bold marks WRM, underlining the strongest C control. AUC uses 27 cases; other metrics use 28. $^{a}$Unfitted/error path. $^{b}$Restricted prompt. Selected fits; unmatched backbones.}
\label{tab:allcomparators}\label{tab:taskadv}
\centering\small\begingroup\setlength{\tabcolsep}{3.2pt}
\begin{tabular}{llrrrrr}\toprule
 & & \multicolumn{2}{c}{Impact} & Violation & \multicolumn{2}{c}{Defect}\\
Run & Instantiated method & MAP & nDCG@10 & F1 & AUC & F1\\\midrule
C & TF-IDF & 0.283 & 0.455 & 0.367 & 0.500 & 0.486\\
C & BM25 transfer & 0.324 & 0.422 & 0.404 & 0.541 & 0.543\\
C & SBERT world 1-NN & 0.519 & 0.678 & 0.559 & 0.736 & 0.710\\
C & GLM-5.2 CoT-labelled & 0.377 & 0.505 & 0.389 & 0.592 & 0.592\\
C & ICL-labelled$^{a}$ & 0.008 & 0.017 & 0.405 & 0.500 & 0.486\\
C & World-text 1-NN (RAG) & 0.551 & 0.706 & 0.586 & 0.759 & 0.735\\
C & Hashed-text MLP & \underline{0.626} & \underline{0.747} & \underline{0.617} & \underline{0.823} & \underline{0.819}\\
C & Type/degree LR & 0.558 & 0.698 & 0.414 & 0.762 & 0.748\\
C & Keyword heuristic$^{a}$ & 0.000 & 0.000 & 0.367 & 0.500 & 0.486\\
C & \textbf{WRM (base)} & \textbf{0.724} & \textbf{0.831} & \textbf{0.693} & \textbf{0.833} & \textbf{0.820}\\
\midrule
\multicolumn{2}{l}{Absolute gain vs. best C control} & \textbf{+0.098} & \textbf{+0.084} & \textbf{+0.076} & \textbf{+0.010} & \textbf{+0.001}\\
\multicolumn{2}{l}{Relative gain vs. best C control} & \textbf{+15.6\%} & \textbf{+11.3\%} & \textbf{+12.2\%} & \textbf{+1.3\%} & \textbf{+0.1\%}\\
\midrule
F & GLM-5.2 CoT-labelled & 0.397 & 0.538 & 0.496 & 0.571 & 0.581\\
F & GLM-5.3 zero-shot$^{b}$ & 0.279 & 0.383 & 0.433 & 0.741 & 0.646\\
F & GLM-5.3 CoT-labelled$^{b}$ & 0.206 & 0.343 & 0.485 & 0.472 & 0.504\\
\bottomrule\end{tabular}\endgroup
\par\smallskip\footnotesize
F repeats SBERT, BM25, RAG and WRM at C scores. Gains use unrounded means. E is a separate per-case export (Table~\ref{tab:paired}); its all-zero GLM-5.3 CoT record is excluded from capability claims.
\end{table}

CodeBERT appears only in external traceability records~\cite{feng2020codebert}. SetFit, Llama pointwise, GLM-5.2 zero-shot, CoT self-consistency, and graph-only GNN are planned or comment-only entries without completed headline results. The implemented decision-agnostic class is not invoked by the evaluation driver. Earlier developmental ablations use a different model/corpus setting and are not matched ablations of this 7B checkpoint. The artifact inventory preserves these distinctions and the task-specific comparisons.

\subsection{Structural strata and inference configuration}
RQ3 partitions the 28 E cases by two input properties: artifact count $|V|$ and directed relation density $\rho=|E|/[|V|(|V|-1)]$. Each split uses the input median, assigning ties to the lower group: 15 artifacts and density 0.08452. Neither threshold uses a model score. Within each group we compare WRM, hashed text, type/degree LR, and world-text transfer. These descriptive partitions reuse fixed predictions and overlap across the two structural properties.

A separate record, \textbf{P} (\texttt{eval\_v8.json}), compares an ICR-enabled WRM with a GLM-5.2 CoT-labelled predictor on 28 recorded cases. Its driver restores refinement and reads the final-step outputs. Its defect decoder accepts every node type, whereas C/F restrict defect output types. P thus evaluates a distinct inference pipeline. Without saved per-case identifiers and an otherwise identical ICR-off run, it cannot establish the isolated effect of refinement.

\begin{figure}[t]
\centering\includegraphics[width=\linewidth]{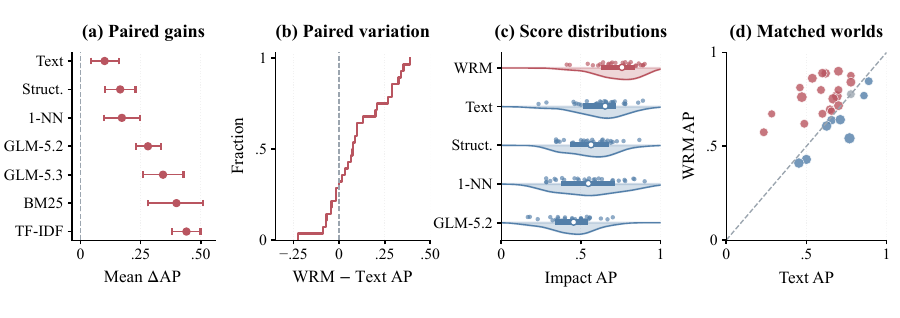}
\caption{Impact-ranking evidence in E ($n=28$). (a) WRM-minus-control means and 95\% paired case-bootstrap intervals (10,000 draws). (b) Empirical distribution of WRM-minus-text AP. (c) Rainclouds retain all cases, median and IQR. (d) Paired scores; point area follows node count, with red/blue/gray marking 19 wins, eight losses and one tie. Text: hashed-text MLP; Struct.: type/degree LR; 1-NN: world-text transfer. GLM labels identify recorded endpoints. Selected fits and unmatched controls limit inference.}
\label{fig:historic}\label{fig:distributions}
\Description{Four aligned panels show paired-effect intervals for seven controls, a cumulative difference distribution, raincloud distributions and matched-world scores.}
\end{figure}

\subsection{Metrics and uncertainty}
Neither ranking metric directly measures analyst effort. Conflict and violation scores average case-wise binary macro-F1, assigning 1.0 when both gold and predictions are empty. Defect AUC (area under the receiver operating characteristic curve) uses 27 cases containing both classes; defect F1 thresholds emitted scores at 0.5 across 28 cases. Percentile intervals use 10,000 case-bootstrap draws with seed 42, paired for differences~\cite{efron1992bootstrap}. They condition on the selected fits and omit model-selection, training-seed, annotation, and scenario uncertainty. Dataset-level comparisons require a different independent unit~\cite{demsar2006}. An AI assistant (Codex) supported analysis and plotting scripts; weighted-sum identities and exact deletion bounds verify the new sensitivity calculations.

For a nonempty gold impact set $G_s$ and a unique ranked list $R_s$, average precision is
\begin{equation}
 \AP_s=\frac{1}{|G_s|}\sum_{j=1}^{|R_s|}\mathbf1[R_{s,j}\in G_s]\,
 \frac{|\{R_{s,1},\ldots,R_{s,j}\}\cap G_s|}{j}.
\end{equation}
Recorded MAP averages $\AP_s$ over cases and uses all annotated target identifiers in $G_s$, including targets outside the current graph. All listed targets count as relevant; annotation weights do not enter AP or nDCG. The eligible-node conversion instead uses $G_s\cap V_s$. It describes the same predictions under a different denominator, not an improvement from a newly trained model.

For binary relevance, normalized discounted cumulative gain at ten (nDCG@10) is:
\begin{equation}
 \mathrm{nDCG@10}(s)=\frac{\sum_{j=1}^{\min(10,|R_s|)}\mathbf1[R_{s,j}\in G_s]/\log_2(j+1)}{\sum_{j=1}^{\min(10,|G_s|)}1/\log_2(j+1)}.
\end{equation}
\section{Results}\label{sec:results}
\subsection{RQ1: Single-step consequence prediction}

\textbf{Impact ranking and constraint scoring.} WRM has the highest recorded mean on all five C metrics (Table~\ref{tab:allcomparators}). Against the strongest control, impact MAP improves by 15.6\%, nDCG@10 by 11.3\%, and violation F1 by 12.2\%. The two ranking measures locate the gain both across relevant positions and near the start of the review list; the violation result adds evidence on identifying constraint risks. Together, they support a useful combination of artifact prioritization and consequence discrimination. Defect gains are smaller (1.3\% AUC; 0.1\% F1), so prioritization and constraint scoring provide the clearest evidence of benefit.

The paired export E permits a closer examination of impact ranking. WRM has MAP 0.724 versus 0.623 for hashed text, with paired mean difference 0.101 and conditional 95\% interval $[0.044,0.159]$. Figure~\ref{fig:historic}(a,b) shows the intervals and the empirical difference distribution. These are E-specific comparisons: hashed text scores 0.626 in C, and GLM-5.2 prompting scores 0.443 in E versus 0.377 in C. The separate records retain those differences rather than treating repeated method names as repeated measurements of a fixed predictor.

\begin{table}[t]
\caption{Direct WRM--control impact comparison in E ($n=28$). Relative gain is $(\MAP_{WRM}-\MAP_b)/\MAP_b$; W/T/L counts higher/equal/lower WRM scores. Bold marks WRM scores and gains. Case-bootstrap intervals condition on selected fits and do not correct checkpoint selection.}
\label{tab:paired}
\centering\small
\begingroup\setlength{\tabcolsep}{3.8pt}
\begin{tabular}{lrrrrrl}
\toprule
Control & $\MAP_b$ & $\MAP_{WRM}$ & $\Delta$ MAP & Rel. gain & W/T/L & 95\% interval\\
\midrule
Hashed-text MLP & 0.623 & \textbf{0.724} & \textbf{+0.101} & \textbf{16.2\%} & 19/1/8 & [0.044,0.159]\\
Type/degree LR & 0.558 & \textbf{0.724} & \textbf{+0.165} & \textbf{29.6\%} & 21/0/7 & [0.102,0.229]\\
World-text 1-NN & 0.551 & \textbf{0.724} & \textbf{+0.172} & \textbf{31.3\%} & 23/0/5 & [0.099,0.248]\\
GLM-5.2 prompt & 0.443 & \textbf{0.724} & \textbf{+0.281} & \textbf{63.4\%} & 28/0/0 & [0.230,0.333]\\
GLM-5.3 zero-shot & 0.380 & \textbf{0.724} & \textbf{+0.343} & \textbf{90.3\%} & 25/0/3 & [0.260,0.429]\\
BM25 transfer & 0.324 & \textbf{0.724} & \textbf{+0.399} & \textbf{123.1\%} & 24/1/3 & [0.279,0.511]\\
TF-IDF & 0.283 & \textbf{0.724} & \textbf{+0.441} & \textbf{155.9\%} & 28/0/0 & [0.379,0.499]\\
\bottomrule\end{tabular}\endgroup
\end{table}

Table~\ref{tab:paired} shows that the ranking gain spans all seven available E controls. WRM improves MAP by 16.2\% over hashed text, 29.6\% over type/degree LR, and 31.3\% over world-text transfer, covering comparisons with lexical, structural, and retrieved-label predictors. It exceeds GLM-5.2 prompting and TF-IDF on all 28 cases, and hashed text on 19, with eight losses and one tie (Figure~\ref{fig:distributions}(d)). This breadth supports an advantage for the recorded WRM scorer while locating individual counterexamples. Backbone and decoding differences prevent isolating graph structure as the cause; RQ4 tests the stability of the observed ordering.

\subsection{RQ2: Candidate, action, and transition validity}
\begin{figure}[t]
\centering\includegraphics[width=\linewidth]{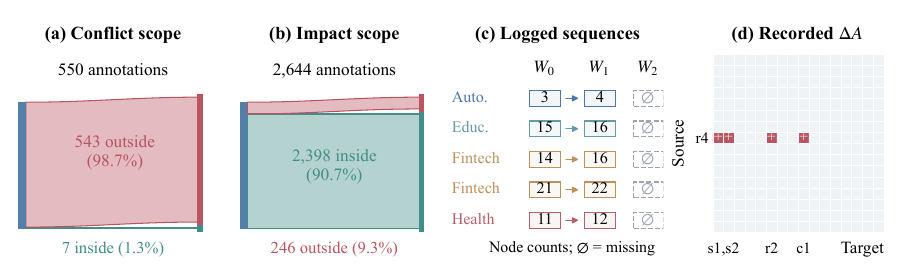}
\caption{Candidate coverage and recorded transition support. (a) Only 7/550 conflicts have both endpoints in the input graph. (b) Current nodes cover 2,398/2,644 impact annotations. Ribbon widths are proportional within each head. (c) Node counts in five two-step trajectories; $\varnothing$ denotes missing successor labels. The other 272 trajectories have one step. (d) The 16-node education edit has four added adjacency entries and no removals; only changed endpoints are labelled. These are corpus records, not WRM rollout predictions.}
\label{fig:coverage}\label{fig:integrity}
\Description{Four panels show conflict and impact annotation flows, five recorded state sequences and the sparse education-case edit matrix.}
\end{figure}

\textbf{Candidate coverage and metric validity.} The mean attainable AP ceiling is 0.89426. Replacing the gold denominator with current-node positives re-expresses WRM's recorded score as eligible AP@10 of 0.81243: each case is multiplied by $|G_s|/|G_s\cap V_s|$. This exact conversion isolates denominator coverage while retaining the same predictions. Conflict scoring has a more severe mismatch. All 56 nominal-test conflicts include an excluded endpoint, so projection onto current-node pairs removes every positive. Empty predictions then receive 1.0 in all 28 cases. Tensorization, decoding, and evaluation therefore exclude newly introduced conflicts (Figure~\ref{fig:coverage}(a,b)).

\textbf{Transition evidence.} Of 282 successor graphs, 276 are empty, including 27/28 nominal-test cases. Only five trajectories have two steps, and none supports horizons three or four (Figure~\ref{fig:integrity}(c,d)). In addition, the relabeling prompt does not freeze the stored decision, leaving action alignment unresolved. These constraints prevent a reliable assessment of successive predicted edits. Node-count preservation alone adds little evidence: the decoder preserves current nodes by construction.

\textbf{Auxiliary capability checks.} Archived deletion, gap, and cascade checks test non-increasing counts, positive-only gap activation, and nonempty output. Stakeholder hits (22/25) lack a precision measure; decision-ranking correlation is $-0.016$ on five cases. Type decodability does not establish causal world knowledge without appropriate controls and interventions~\cite{hewitt2019,li2022world}. An earlier developmental ablation also has higher no-graph MAP (0.610) than full-model MAP (0.589). These records are retained as limits on capability and mechanism claims, rather than counted as additional evidence for the headline checkpoint.

\subsection{RQ3: World structure and inference configuration}
\textbf{Comparisons across world structures.} WRM leads all three controls in mean AP in each of the four input-defined strata (Table~\ref{tab:structure}). For worlds with at most 15 nodes, WRM scores 0.752 versus 0.615 for hashed text; in larger worlds, 0.686 versus 0.634. Positive differences of 0.137 and 0.053 show that the observed advantage spans both size groups, with a larger margin in smaller worlds. Candidate-set ceilings (0.897 versus 0.891) and mean impact-set sizes (9.25 versus 9.33) are similar. Candidate coverage alone therefore does not explain the difference in margins.

\begin{table}[t]
\caption{Direct impact comparison across input-defined strata of E. Text: hashed-text MLP; Struct.: type/degree LR; 1-NN: world-text transfer. Intervals are paired case-bootstrap intervals for WRM minus Text. The two partitions overlap; this is exploratory selected-case analysis.}
\label{tab:structure}\centering\small
\begingroup\setlength{\tabcolsep}{3pt}
\begin{tabular}{lrrrrrrl}\toprule
World stratum & $n$ & WRM & Text & Struct. & 1-NN & $\Delta$ Text & 95\% interval\\\midrule
$|V|\leq15$ & 16 & \textbf{0.752} & 0.615 & 0.513 & 0.554 & +0.137 & [0.066,0.210]\\
$|V|>15$ & 12 & \textbf{0.686} & 0.634 & 0.619 & 0.548 & +0.053 & [-0.032,0.140]\\
$\rho\leq0.08452$ & 14 & \textbf{0.699} & 0.630 & 0.580 & 0.512 & +0.069 & [-0.009,0.147]\\
$\rho>0.08452$ & 14 & \textbf{0.748} & 0.616 & 0.537 & 0.590 & +0.132 & [0.054,0.214]\\
\bottomrule\end{tabular}\endgroup\end{table}

Density exhibits a different pattern: the WRM-minus-text difference is 0.132 in the denser half and 0.069 in the sparser half. This is compatible with relational context becoming useful when more links are available, but node count, domain, label prevalence, and relation types vary jointly. Padding behavior also complicates graph-size interpretation (Section~\ref{sec:method}). The intervals for larger and sparser worlds include zero. Thus these strata locate variation in the observed gain; matched graph interventions are needed to attribute it to particular structural properties.

\textbf{The ICR-enabled pipeline.} In P (Table~\ref{tab:pilot}), WRM has impact MAP 0.724 versus 0.524, nDCG@10 0.831 versus 0.668, and violation F1 0.732 versus 0.638: relative gains of 38.2\%, 24.4\%, and 14.7\%. Its larger defect-score differences belong to P's all-type decoder. This supports a pipeline-level comparison, not an isolated refinement effect, because no matched ICR-off run is available. The identical WRM MAP in C and P is consistent with the fixed-world ordering property in Section~\ref{sec:advantage}; refinement need not change impact order even when other outputs change.

\begin{table}[t]
\caption{Reported comparison P: ICR-enabled WRM versus its recorded GLM-5.2 comparator. AUC uses 27 cases; other metrics use 28. P is a separate pipeline record with no saved case identifiers. Its all-type defect decoder differs from C/F; rows do not constitute an ICR ablation.}
\label{tab:pilot}\centering\small
\begingroup\setlength{\tabcolsep}{4pt}
\begin{tabular}{lrrrrr}\toprule
 & \multicolumn{2}{c}{Impact} & Violation & \multicolumn{2}{c}{Defect}\\
Recorded system & MAP & nDCG@10 & F1 & AUC & F1\\\midrule
GLM-5.2 CoT-labelled & 0.524 & 0.668 & 0.638 & 0.597 & 0.569\\
\textbf{WRM with ICR} & \textbf{0.724} & \textbf{0.831} & \textbf{0.732} & \textbf{0.860} & \textbf{0.829}\\
\midrule
Absolute difference & +0.200 & +0.163 & +0.094 & +0.263 & +0.260\\
\bottomrule\end{tabular}\endgroup\end{table}

\textbf{Expanded-corpus comparisons.} Table~\ref{tab:expanded} extends the comparison to four 47-case settings. WRM exceeds the strongest completed control, hashed text, in all 20 setting--metric means, with gains of 15.6--29.1\% in MAP and 9.4--21.8\% in nDCG@10. The ranking improvement is therefore accompanied by higher reported violation and defect means in every setting, rather than an aggregate accuracy trade-off between heads. Magnitudes remain task-dependent: violation F1 rises by only 0.004 under retrieval training. The componentwise ordering survives any nonnegative weighting of the five means with positive total weight, but defines neither per-case dominance nor a validated engineering utility. Full control and failure records remain in the artifact.

\begin{table}[t]
\caption{Comparison across annotation and training settings ($n=47$ per setting). Cells show \textbf{WRM} / strongest completed control (hashed-text MLP throughout); TF-IDF and type/degree LR are also evaluated. All settings evaluate the base scorer, reuse selection cases, and are not independent replications or component ablations.}
\label{tab:expanded}\centering\small
\begingroup\setlength{\tabcolsep}{4pt}
\begin{tabular}{lccccc}\toprule
& \multicolumn{2}{c}{Impact} & Violation & \multicolumn{2}{c}{Defect}\\
Setting & MAP & nDCG@10 & F1 & AUC & F1\\\midrule
Original labels & \textbf{0.655} / 0.508 & \textbf{0.788} / 0.646 & \textbf{0.649} / 0.638 & \textbf{0.762} / 0.746 & \textbf{0.771} / 0.746\\
Revised labels & \textbf{0.717} / 0.598 & \textbf{0.828} / 0.734 & \textbf{0.668} / 0.615 & \textbf{0.860} / 0.834 & \textbf{0.839} / 0.824\\
Retrieval training & \textbf{0.696} / 0.596 & \textbf{0.808} / 0.733 & \textbf{0.619} / 0.615 & \textbf{0.857} / 0.832 & \textbf{0.843} / 0.824\\
Edge modulation & \textbf{0.703} / 0.608 & \textbf{0.811} / 0.742 & \textbf{0.707} / 0.615 & \textbf{0.873} / 0.833 & \textbf{0.849} / 0.826\\
\midrule
Rel. gain range & +15.6--29.1\% & +9.4--21.8\% & +0.7--15.0\% & +2.1--4.9\% & +1.8--3.3\%\\
\bottomrule\end{tabular}\endgroup
\end{table}

The settings share input graphs and splits. Revised labels change all 47 test annotations; retrieval training reuses that corpus, while edge modulation changes only three non-test labels. Case IDs and exact input graphs do not overlap the primary corpus, although scenario templates may recur. These comparisons broaden descriptive evidence without providing independent replications or a scaling experiment. Aggregate-only records cannot support reconstructed paired intervals.

\subsection{RQ4: Robustness across requirement worlds}
\textbf{Lower-tail performance.} Following distributional evaluation practice~\cite{agarwal2021precipice}, we examine all 28 E scores for each method. For method $b$, order its scores as $a_{b,(1)}\leq\cdots\leq a_{b,(28)}$ and define
\begin{equation}
L_m(b)=\frac{1}{m}\sum_{j=1}^{m}a_{b,(j)},\qquad m=1,\ldots,28.
\end{equation}
Thus $L_7$ averages the lowest quarter, $L_{14}$ the lowest half, and $L_{28}$ recovers MAP. Each method selects its own low-scoring cases, so these summaries compare distributions rather than one shared difficult subset. Table~\ref{tab:tailrobustness} shows WRM means of 0.546 and 0.619 for the lowest quarter and half, versus 0.413 and 0.508 for hashed text. The differences are 0.134 (32.4\%) and 0.111 (21.9\%); the medians are 0.758 and 0.652. The mean advantage therefore accompanies stronger lower-tail outcomes, not merely a few high scores.

\begin{table}[t]
\caption{Impact-ranking robustness of 28 selected E fits. Macro-MAP gives each of six domains equal weight. $L_7$ and $L_{14}$ average each method's lowest seven and fourteen AP scores; these subsets can differ. Higher AP/MAP and lower rank are better. Bold marks WRM; underlining identifies the strongest control.}
\label{tab:tailrobustness}\centering\small
\begingroup\setlength{\tabcolsep}{5pt}
\begin{tabular}{lrrrrr}\toprule
 & \multicolumn{4}{c}{Impact AP/MAP ($\uparrow$)} & Rank ($\downarrow$)\\
Method & Macro-MAP & $L_7$ & $L_{14}$ & Median & Mean\\\midrule
\textbf{WRM} & \textbf{0.734} & \textbf{0.546} & \textbf{0.619} & \textbf{0.758} & \textbf{1.964}\\
Hashed-text MLP & \underline{0.648} & \underline{0.413} & \underline{0.508} & \underline{0.652} & \underline{2.982}\\
Type/degree LR & 0.564 & 0.370 & 0.443 & 0.563 & 3.714\\
World-text 1-NN & 0.584 & 0.316 & 0.402 & 0.546 & 4.000\\
GLM-5.2 prompt & 0.453 & 0.271 & 0.337 & 0.454 & 5.339\\
GLM-5.3 zero-shot & 0.364 & 0.135 & 0.232 & 0.412 & 5.518\\
BM25 transfer & 0.314 & 0.000 & 0.070 & 0.300 & 5.821\\
TF-IDF & 0.280 & 0.150 & 0.189 & 0.255 & 6.661\\
\bottomrule\end{tabular}\endgroup
\end{table}

\begin{figure}[t]
\centering\includegraphics[width=\linewidth]{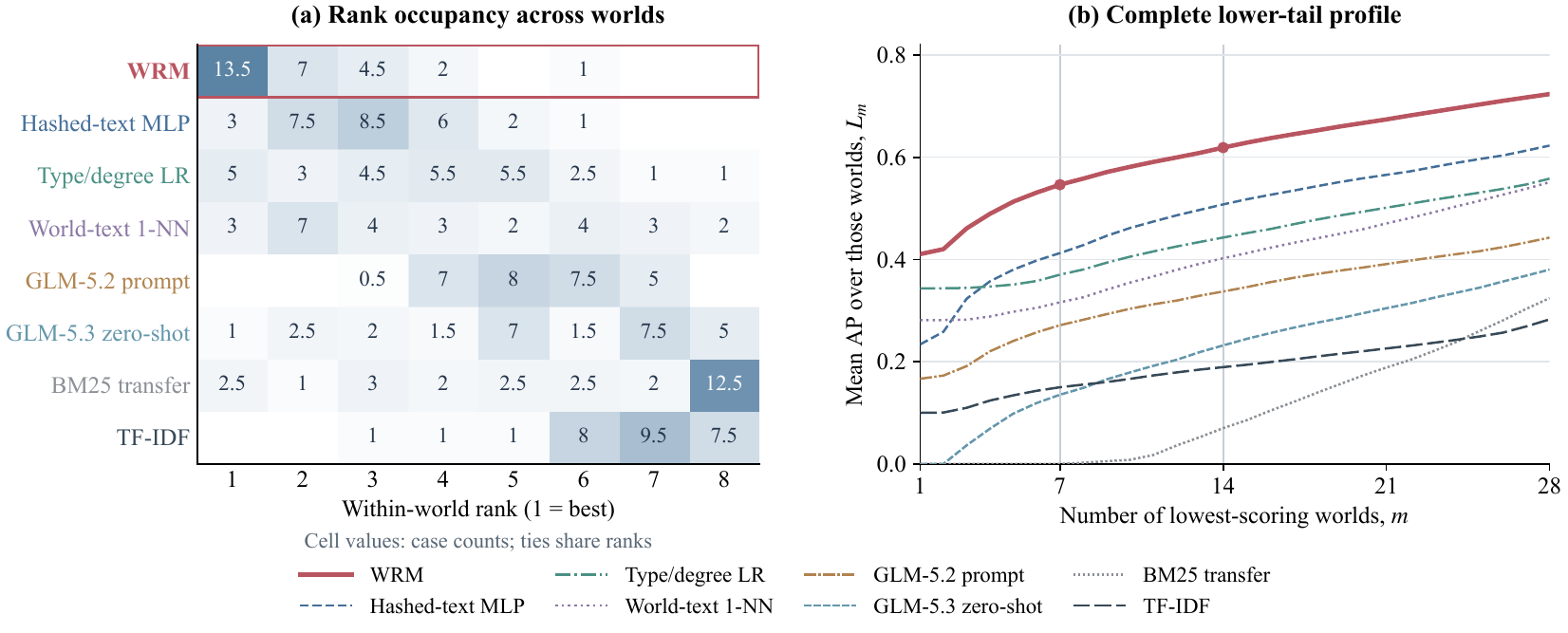}
\caption{Robustness of the recorded impact scores in E. (a) Rank occupancy over all 28 worlds; tied methods split their occupied rank positions. (b) Mean of each method's $m$ lowest AP scores, for every $m=1,\ldots,28$. WRM leads all seven controls at every displayed tail size. Curves reuse selected fits, and method-specific tails need not contain the same worlds.}
\label{fig:tailrobustness}
\Description{A rank-occupancy heatmap and complete lower-tail score curves distinguish case-wise ranking consistency from distributional lower-tail performance.}
\end{figure}

\textbf{Within-world consistency.} Mean rank is 1.96 for WRM versus 2.98 for hashed text; top-three counts are 25/28 versus 18/28, using average ranks for ties. These cohort summaries coexist with eight individual losses and do not establish population dominance. The all-zero GLM-5.3 CoT export is excluded because successful execution is unverified; its raw record remains in the artifact.

\textbf{Composition sensitivity.} Equal domain weights yield MAP 0.734 for WRM versus 0.648 for hashed text, a 13.3\% gain (Table~\ref{tab:tailrobustness}). Removing each domain in turn leaves WRM above all seven controls; the gap to hashed text ranges from 0.079 to 0.116. WRM does not lead in every domain: the strongest controls win in e-commerce and fintech. Sorting the 28 paired gains gives exact deletion bounds: removing any ten worlds leaves a positive mean advantage over hashed text; reversal requires at least 11 deletions. These checks reuse scores and cannot establish domain transfer.

\section{Discussion}\label{sec:discussion}
\subsection{Interpreting the ranking advantage}
The results support complementary aspects of WRM's recorded advantage: higher mean impact and constraint scores, stronger lower-tail outcomes, and stability to domain weighting and case deletion. The expanded-corpus settings add evidence beyond subdivisions of E, subject to their shared data and selection protocol. Proposition~\ref{prop:tailweights} further shows why the complete lower-tail profile matters: its ordering survives every weighting scheme that emphasizes lower-scoring outcomes.

\begin{proposition}[Lower-tail weighting]\label{prop:tailweights}
Let $a_{b,(1)}\leq\cdots\leq a_{b,(n)}$ be method $b$'s sorted AP scores. If $L_m(\mathrm{WRM})>L_m(b)$ for every $m=1,\ldots,n$, then $\sum_j w_j a_{\mathrm{WRM},(j)}>\sum_j w_j a_{b,(j)}$ for any $w_1\geq\cdots\geq w_n\geq0$ with $\sum_jw_j=1$.
\end{proposition}
\begin{proof}
Set $w_{n+1}=0$. Rearranging the finite sum gives
\begin{equation}
\sum_{j=1}^{n}w_j a_{b,(j)}=\sum_{m=1}^{n}m(w_m-w_{m+1})L_m(b).\label{eq:tailweights}
\end{equation}
The coefficients on the right are nonnegative and sum to one. Every WRM-minus-control tail difference is positive, so their weighted sum is positive.
\end{proof}

With $n=28$, Figure~\ref{fig:tailrobustness} satisfies the condition for all seven included E controls. Uniform weights recover MAP; weights restricted to the lowest seven or fourteen scores recover the reported tail means. Against hashed text, the smallest gap over this weight class is 0.10076, attained by uniform weighting over all 28 scores. Thus the ordering holds across every nonincreasing rank-weight schedule on these recorded distributions. The schedules weight each method's own ordered outcomes, not the same named worlds, so this result coexists with eight paired losses and does not imply population dominance.

A plausible mechanism is that typed context helps identify consequences carried by artifact relations beyond textual similarity. Improvements over text, degree-only, and label-transfer controls are compatible with this explanation, but unmatched backbones prevent causal attribution. A decisive test would hold node text, model capacity, and training protocol fixed while changing verification or constraint relations, then independently annotate the consequences. Text-competitive cases belong in that study: they reveal when relational modeling must add information beyond a strong language representation.

For engineering use, WRM's common identifiers could connect a proposed edit, a ranked artifact, and its associated constraint and stakeholder context in one review interface. The examination example illustrates the benefit of keeping these views attached to the same artifacts. Relation paths would supply inspectable context, not validated causal explanations. Measuring accepted suggestions, missed impacts, and review time would establish whether better ranking translates into better change decisions. Such a study should match candidate universes and review budgets, reporting retrieved impacts and missed constraint risks, including excluded targets. Better mean ranking can still leave an important constraint below the cutoff, so consequence-specific coverage must accompany the rank score.

\subsection{A validation path for requirement-world prediction}
Algorithm~\ref{alg:contract} specifies the evaluation conditions needed to move from contextual scoring to requirement-world prediction. An addition needs materialized candidate nodes, a modification needs old and new text, and a deletion needs removal semantics. Independent edit tests separate decoder guarantees from learned accuracy.

\begin{algorithm}[t]
\caption{Evaluation protocol for a requirement-world study}\label{alg:contract}
\begin{algorithmic}[1]
\Require Cases $(W,d,W',C)$ and head-specific universes $U_h$
\State Validate node identifiers, edge endpoints, types, and decision payload
\State Check that repeated annotations use the same $(W,d)$ for each case
\State Construct $U_h$ explicitly, materializing candidate nodes if required
\For{each case and consequence head $h$}
  \State Compute positive-label coverage $|C_h\cap U_h|/|C_h|$ where defined
  \State Retain out-of-universe labels as coverage failures or reject the case
\EndFor
\State Assign project/scenario groups before fitting or label-based tuning
\State Select the model and decoding thresholds using training/validation only
\State Evaluate untouched cases; log invalid outputs separately from abstentions
\State Score full graph edits and aligned consequences at each observed horizon
\State Report candidate coverage, prevalence, denominators, and uncertainty
\end{algorithmic}
\end{algorithm}

The next comparison should hold the Qwen snapshot, adaptation, input budget, candidate universe, and selection rule fixed across text-only and graph variants. Edge permutation and removal of graph attention, typed propagation, stakeholder grounding, or refinement test distinct mechanisms. Labels must concern the exact stored decision, retain disagreements, include varied operations and negative cases, and undergo blinded human review. Transition evaluation additionally needs nonempty successors and sufficiently long trajectories. Copy-state and deterministic-edit controls should anchor per-horizon node, edge, text, and consequence evaluation.

Action sensitivity can be tested by pairing two decisions on the same world with independently annotated differences in impact priority. For artifact pairs whose correct order reverses, compare the margins $s_i(W,d)-s_j(W,d)$ across the two decisions. Equation~\ref{eq:factorization} predicts unchanged margins for the current linear head, including within-world refinement. A new variant therefore needs node--decision interactions, through the representation or the readout. Correct reversals and erroneous reversals on unchanged pairs should both be reported, separating useful sensitivity from instability. This proposed experiment tests a capability that the present MAP comparison cannot establish.

\section{Threats to Validity}\label{sec:threats}
\textbf{Construct validity.} The labels are generated judgments about synthetic requirements, not deployed-software outcomes. Relabeling may attach consequences to a different decision, and missing raw decisions prevent estimating that error. Type decoding, nonempty cascades, and preserved node counts are weaker constructs than causal world knowledge, correct edits, or stakeholder agreement. The head-specific candidate analysis makes these distinctions explicit.

\textbf{Internal validity.} The scored cases selected WRM checkpoints. Backbone, preprocessing, input-budget, and decoding differences prevent component-level attribution. C/F/E, P, and the expanded-corpus records remain separate; P and the four expanded settings lack complete per-case prediction manifests. Structural strata reuse E and are exploratory. These choices preserve the provenance of each comparison but do not remove selection bias or establish a matched ablation.

\textbf{External validity.} The six domains share generation policies, and 279/282 primary decisions are additions. Small graphs, short descriptions, repeated scenario families, and ten education cases limit transfer to industrial repositories, other operations, languages, and project histories. A larger generated corpus does not by itself address that limitation. Archived PROMISE classification also favors generic Qwen over WRM features (weighted F1 0.833 versus 0.758); the complete external comparison remains in the artifact.

\textbf{Statistical validity and reproducibility.} Case-bootstrap intervals omit selection, seed, annotation, and scenario uncertainty. Lower-tail dominance concerns the selected cohort and allows individual-case losses. Hashes and analysis scripts reproduce the reported arithmetic, not historical training or closed-service behavior. Full neural reproduction requires the checkpoint, model snapshot, dependencies, and endpoint access; pretraining contamination remains unassessed.

\section{Related Work}\label{sec:related}
\textbf{Requirements, traceability, and change analysis.} Foundational RE separates domain knowledge, goals, and specifications~\cite{zave1997,lamsweerde2009}; traceability connects them across the artifact lifecycle~\cite{gotel1994traceability,cleland2012}. Event-based maintenance~\cite{cleland2003event} and formal relation semantics~\cite{goknil2011semantics,goknil2016impact} support change analysis. Trace recovery uses information retrieval and semantic learning~\cite{borg2014traceability,guo2017traceability}, motivating our TF-IDF and sentence-embedding controls~\cite{manning2008,reimers2019}. Recent approaches emphasize prompt and demonstration design in TraceLLM~\cite{alturayeif2026tracellm} and LLM-based data augmentation with an enhanced encoder~\cite{zhang2026traceability}. NLP and LLM studies also address classification and elicitation~\cite{zhao2020nlp,ronanki2023,alhoshan2025}. WRM takes these artifact relations as supplied context and learns several consequence scores for a proposed edit.

\textbf{Structured representation and shared prediction.} Transformers, graph attention, and structural biases combine textual and relational information~\cite{vaswani2017attention,velickovic2018gat,ying2021graphormer}. Relational inductive biases and relation-specific graph learning preserve edge meaning~\cite{battaglia2018relational,schlichtkrull2017}; multitask learning and listwise objectives support shared supervision and ranking~\cite{caruana1997,cao2007}. WRM combines these ingredients around persistent artifact identities, connecting several consequence scores to one engineering state. The contribution is their requirement-world formulation and evaluated scoring behavior; matched interventions remain necessary to attribute gains to individual components.

\textbf{World models and predictive representations.} World Models, Dreamer, and TD-MPC2 learn predictive representations for control~\cite{ha2018world,hafner2019dream,hafner2025dreamer,hansen2024tdmpc2}. I-JEPA predicts in representation space~\cite{assran2023ijepa}; V-JEPA~2 extends video representation learning with an action-conditioned model for robotic planning~\cite{assran2025vjepa2}. WRM applies a state--edit--successor formulation to engineering artifacts, with its strongest evidence in consequence scoring. Successor validity includes artifact identity, textual revisions, and changed relations. Latent agreement cannot establish a correct edit: an addition requires a new artifact. Separate graph and consequence outputs make these properties testable. Emergent-world studies use decoding and interventions~\cite{li2022world}; WRM receives types explicitly, so decodability alone cannot establish learned world knowledge~\cite{hewitt2019}.

\textbf{Evaluation scope.} Formal RE research proposes combining LLM support with formal checks on generated artifacts~\cite{ferrari2025formal}; cross-task RE evaluation finds that model performance depends on the activity~\cite{dabrowski2026crosstask}. WRM therefore needs distinct evidence for scoring, edit validity, and action sensitivity. Distributional evaluation complements means~\cite{agarwal2021precipice}, while selection bias, leakage, contamination, and dependent observations require separate remedies~\cite{cawley2010,kapoor2023,dong2024,demsar2006}. Our candidate and readout analyses clarify whether targets are representable and scores can depend on the decision.

\section{Conclusion}\label{sec:conclusion}
World Requirement Model couples typed artifact context with shared consequence readouts, making multiple review questions addressable within one requirement world. On the paired cohort, WRM improves impact MAP by 16.2\% and lowest-quarter mean AP by 32.4\%; equal domain weights retain a 13.3\% gain, and any ten-case deletion preserves the mean advantage over hashed text. Four expanded-corpus settings show MAP gains of 15.6--29.1\%. These results from selected checkpoints support stronger contextual prioritization. Shared artifact identities make predicted priorities and risks inspectable in their engineering context. Broader world prediction still requires candidate closure, aligned decisions, independent assessment, and action-sensitive readouts. Industrial studies should test whether these ranking gains reduce review effort.

\section*{Data Availability}
We have assembled a replication package containing the frozen primary corpus and expanded-corpus snapshots, archived comparator results, per-case impact AP scores, analysis scripts, retained configuration and driver snapshots, figure assets, and provenance audits. File hashes and mappings between published setting names and source records support traceability. The scripts recompute the reported aggregates, paired score comparisons, candidate coverage, domain-balanced MAP, and case-deletion bounds, and regenerate derived tables and statistical figures from the retained records. These materials support verification of the reported analyses without rerunning neural training. Reproducing model inference additionally requires the original checkpoint, exact Qwen snapshot, compatible software environment, and access to any closed-model services used by the relevant comparator. Aggregate-only records lack complete per-case prediction manifests and cannot support reconstructed paired intervals. The package is currently maintained locally. Upon acceptance, we will publicly release the data and open-source the code in the replication package. This version provides no repository access link.

\bibliographystyle{ACM-Reference-Format}
\bibliography{wrm_references_new}
\end{document}